\documentclass{jgaa-art}

\usepackage{amsmath,amssymb}
\usepackage{algorithm,algorithmic,color}
\usepackage{enumerate,graphicx,hyperref}
\usepackage{cite}

\newtheorem{corollary}{Corollary}
\newtheorem{definition}{Definition}
\newtheorem{observation}{Observation}
\newtheorem{example}{Example}

\title{Metric Dimension Parameterized\\ by Max Leaf Number}
\Ack{This research was supported in part by the National Science Foundation under grants 0830403 and 1217322, and by the Office of Naval Research under MURI grant N00014-08-1-1015.}
\author{David Eppstein}{eppstein@uci.edu}
\affiliation{Department of Computer Science, University of California, Irvine}

\begin{document}
\maketitle

\begin{abstract}
The \emph{metric dimension} of a graph is the size of the smallest set of vertices whose distances distinguish all pairs of vertices in the graph. We show that this graph invariant may be calculated by an algorithm whose running time is linear in the input graph size, added to a function of the largest possible number of leaves in a spanning tree of the graph.
\end{abstract}

\section{Introduction}

Since its initial formulation, the theory of parameterized complexity has had great success in developing algorithms for $\mathsf{NP}$-hard problems that are general enough to handle all inputs, that are fast on inputs of low complexity (as measured by the parameter of interest), and that degrade gracefully as this parameter increases. For instance, by Courcelle's theorem, a large number of graph properties have fixed-parameter tractable algorithms when parameterized by treewidth~\cite{Cou-IC-90}; these algorithms have running time bounds that are linear in the size of the graph, multiplied by non-polynomial functions of the treewidth. An even larger class of problems (essentially, all monotone graph properties) have fixed-parameter tractable algorithms when parameterized by tree-depth~\cite{NesOss-12}.
Nevertheless, for some important graph problems and parameters, fixed-parameter tractable algorithms with these parameters are unknown or (if standard complexity-theoretic assumptions hold) provably do not exist.  One example of this phenomenon is given by the \emph{metric dimension} of a given graph~\cite{HarMel-AC-76}.

\begin{definition}
A  \emph{locating set} (or \emph{metric basis}) for a graph $G$ is a set $S$ of vertices with the property that, for every two vertices $u$ and $v$ in $G$, there exists a vertex $w\in S$ such that $u$ and $v$ have different distances to $w$. The \emph{metric dimension} of~$G$ is the minimum cardinality of a locating set for~$G$.
\end{definition}

Thus, the locating set gives a set of landmarks that can be used for unambiguous navigation in~$G$, and the metric dimension counts the number of landmarks that are necessary for this purpose~\cite{KhuRagRos-DAM-96}.
The graphs for which the metric dimension is bounded may be recognized in polynomial time, by an obvious brute-force search algorithm that tests whether each tuple with the given size bound is a locating set. Generalizing an algorithm for metric dimension in trees~\cite{HarMel-AC-76}, the metric dimension may also be computed in polynomial time for graphs of bounded \emph{cyclomatic number} (the minimum number of edges the removal of which breaks all cycles)~\cite{EpsLevWoe-WG-12}. However, the exponents of these algorithms depend on their parameters, so they are not fixed-parameter tractable, and the problem does not seem to fit into the standard classes of problems that may be solved efficiently for graphs of bounded treewidth~\cite{DiaPotSer-ESA-12} or tree-depth. Additionally, the metric dimension of a graph is complete for $\mathsf{W}[2]$~\cite{HarNic-CCC-13}, again implying that it is unlikely to be fixed-parameter tractable for its natural parameter. This negative result implies that, in order to find fixed-parameter tractable algorithms for this problem, we must search for weaker parameters that better distinguish the easy instances of these problems from the hard ones.

\begin{figure}[t]
\centering\includegraphics[width=4in]{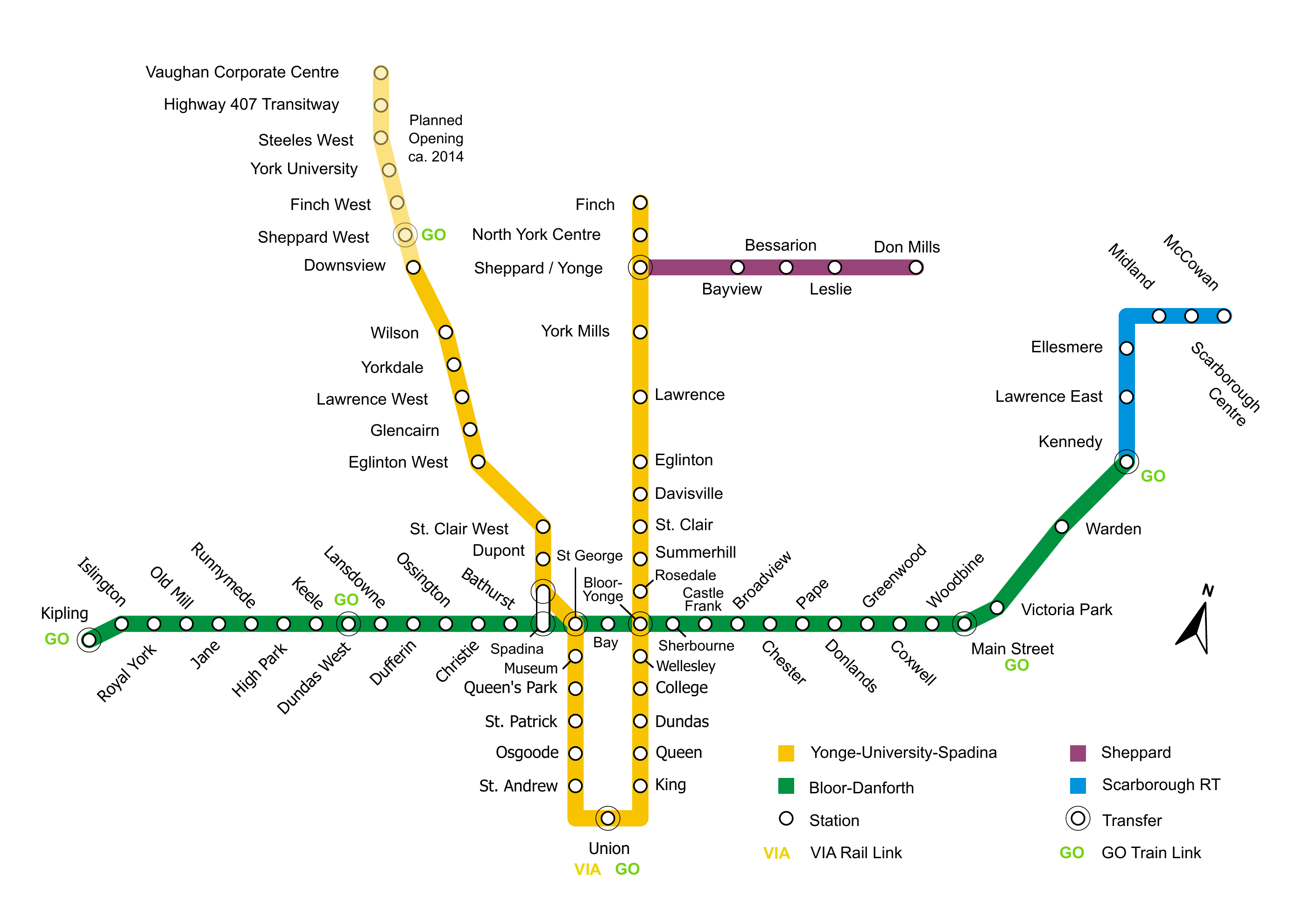}
\caption{The Toronto TTC subway system, a graph with 75 vertices, max leaf number~7, and 8~branches. Public domain image by Paulshannon from \href{http://commons.wikimedia.org/wiki/File:TTCsubwayRTmap-2007.svg}{Wikimedia commons}.}
\label{fig:TTC}
\end{figure}

In this paper, we find such a result, parameterized by the \emph{max leaf number} of a graph.  Our algorithms are particularly efficient for graphs with many degree-two vertices and few vertices of other degrees, which are common for instance in subway and train systems (Figure~\ref{fig:TTC}).

\begin{definition}
The \emph{max leaf number} of a connected graph~$G$ is the maximum, over all spanning trees of~$G$, of the number of leaves in the spanning tree.
\end{definition}

The max leaf number of~$G$ can equivalently be defined as the maximum number of leaves in a star $K_{1,\ell}$ that is a minor of~$G$, because contracting the interior edges of a tree with $\ell$ leaves leads to a star minor~\cite{FelLan-SJDM-92}. It also equals the maximum degree of a minor of~$G$. Because of these equivalent definitions, the max leaf number is minor-monotone. 
Testing whether the max leaf number is at most a given threshold is $\mathsf{NP}$-complete~\cite[ND4, p. 206]{GarJoh-79} but the max leaf number is fixed-parameter tractable with its natural parameter~\cite{FelLan-SJDM-92} and parameterized algorithms for computing it have been the subject of extensive algorithmic research (see e.g. Fernau \emph{et al.}~\cite{FerJoaDie-TCS-11} and their references).
After an initial investigation by Fellows et al.~\cite{FelLokNus-TCS-09},
the max leaf number has by now become one of the standard choices for parameterizing algorithms for other graph problems~\cite{AdiChiSau-ISAAC-10,BodJanKra-TCS-13,FelHerRos-ESA-13,Lam-Algo-12,Lam-LMCS-14}.

\section{Max leaf number versus branches}

Rather than parameterizing our algorithms directly by the max leaf number, it will be convenient for us to instead use a different but (as we prove) functionally equivalent parameter, the number of \emph{branches} in the given graph.

\begin{definition}
A \emph{branch} of a graph $G$ is a maximal path or cycle in which every internal vertex of the path has degree two in $G$. A vertex~$v$ \emph{belongs} to a branch if $v$ is incident to an edge of the branch and it is not incident to edges of any other branches.
\end{definition}

\begin{lemma}
\label{lem:branch-upper-bound-by-leaf}
In any connected graph with max leaf number $\ell$, there can be at most $O(\ell^2)$ branches.
\end{lemma}

\begin{proof}
We prove in the opposite direction that if there are $b$ branches then there is a tree with $\Omega(\sqrt b)$ leaves.
So, suppose that we have a connected graph $G$ with $b$ branches.
We partition into cases:
\begin{itemize}
\item Suppose that at least $\sqrt b$ of these branches end in a degree-one vertex.
Then contracting all the other branches leaves a tree with at least $\sqrt b$ leaves.
\item If we are not in the previous case, form a graph $G'$ (the 2-core of~$G$) by recursively removing all degree-one branches from $G$. There can be $O(\sqrt b)$ removed branches, and each removed branch may cause two remaining branches to merge, so $G'$ has $b-O(\sqrt b)$ branches. Let $t$ be the number of vertices in $G'$ of degree three or more. If $t=O(\sqrt b)$, then at least one of these vertices must have $\Omega(\sqrt b)$ branches incident to it, giving a tree with $\Omega(\sqrt b)$ leaves.
\item In the remaining case, we have $b-O(\sqrt b)$ branches in $G'$ and $t=\Omega(\sqrt b)$ vertices of degree three or more. Contracting each branch of $G'$ to a single edge forms a graph with $t=\Omega(\sqrt b)$ vertices, each of which has degree at least three. A classical theorem of Kleitman and West~\cite{KleWes-SJDM-91} implies that the contracted graph has a tree with $\Omega(\sqrt b)$ leaves. Undoing the contraction results in a tree with the same number of leaves in $G$ itself.
\end{itemize}
Thus in every case $G$ has max leaf number $\Omega(\sqrt b)$
\end{proof}

\begin{example}
A complete graph $K_n$ has max leaf number $n-1$ and  $n(n-1)/2$ branches. This example shows that \autoref{lem:branch-upper-bound-by-leaf} is asymptotically tight.
\end{example} 

\begin{lemma}
\label{lem:leaf-upper-bound-by-branch}
Every connected graph with $b>0$ branches has max leaf number at most $2b$.
\end{lemma}

\begin{proof}
If a graph $G$ has max leaf number $\ell$, then it has a tree $T$ with $\ell$ leaves, and therefore (if $\ell>2$) it has at least $\ell$ branches. Each branch of $G$ can give rise to at most two branches of~$T$, so the number $b$ of branches of~$G$ obeys the inequality $2b\ge \ell$. The case when $\ell\le 2$ is even easier, for  in this case the graph must be a path with $b=1$ and $\ell=2$.
\end{proof}

\begin{corollary}
\label{cor:equiv-params}
Any graph algorithm that is fixed-parameter tractable for max leaf number is fixed-parameter tractable for the number of branches, and vice versa.
\end{corollary}

\begin{proof}
This follows immediately from \autoref{lem:branch-upper-bound-by-leaf} and \autoref{lem:leaf-upper-bound-by-branch}.
\end{proof}

\section{Metric dimension}

With these preliminaries about numbers of branches in hand, we are ready to start describing our algorithm for the metric dimension.

\subsection{Indistinct sets}

In a graph with a small number of branches, a single vertex in a locating set will necessary distinguish most of the pairs of vertices in the graph.

\begin{lemma}
\label{lem:monotonic}
Let $G$ be a graph, $B$ be a branch of $G$, and $s$ be any vertex of $G$. Then $B$ may be partitioned into at most three contiguous paths within which the distance from $s$ is monotonic.
\end{lemma}

\begin{proof}
If $s$ is not within $B$, then let $v$ be the point of $B$ where distance from $s$ is largest. Then splitting $B$ into two paths at $v$ necessarily gives two contiguous paths on which the distance is monotonic: neither path can contain a local minimum of distance, because the only possible such point within a path is $s$ itself, and neither path can contain a local maximum, because there would have to be a local minimum between any local maximum and~$v$.

If $s$ is within $B$, then split $B$ into three paths at $s$ and at the point $v$ where distance from $s$ is largest. The two paths ending at $v$ are monotone for the same reason as before. The third path, from $s$ to the other endpoint $w$ of $B$, must also be monotone. For, if it had a local maximum at a point $u$, then the shortest path from $s$ to $w$ would be shorter than both of the paths from $s$ to $u$ and $s$ to $v$, and would therefore have to avoid both $u$ and $v$, but there is no path in $G$ from $s$ to $w$ that avoids both $u$ and $v$.
\end{proof}

\begin{definition}
Let $G$ be a graph, with $A$ and $B$ being two of its branches, and let $s$ be a vertex in a locating set for $G$.
Then the \emph{indistinct set} for $s$, $A$, and $B$ is defined to be the set of pairs  of vertices $(a,b)$ with $a\in A$ and $b\in B$ with $d(s,a)=d(s,b)$. We do not require $A$ and $B$ to be distinct, so $A=B$ is allowed in this definition.
\end{definition}

\begin{lemma}
Let $G$ be a graph, with $A$ and $B$ being two of its branches, and let $s$ be a vertex in a locating set for $G$.
Then the indistinct set for $s$, $A$, and $B$ has size $O(\min(|A|,|B|))$ .
\end{lemma}

\begin{proof}
By \autoref{lem:monotonic} the vertices in $A$ and in $B$ may be divided into at most three paths per branch, within which the distance from $s$ is monotonic.
Therefore, there are $O(1)$ points in both $A$ and $B$ that have a given distance $d$ from $s$,
and only $O(1)$ pairs of one point from $A$ and one point from $B$ that both have this distance.
The total number of pairs that are not distinguished is the sum of this $O(1)$ bound over
the at most $\min(|A|,|B|)$ different distances that need to be distinguished.
\end{proof}

\begin{figure}[t]
\centering\includegraphics[width=4.5in]{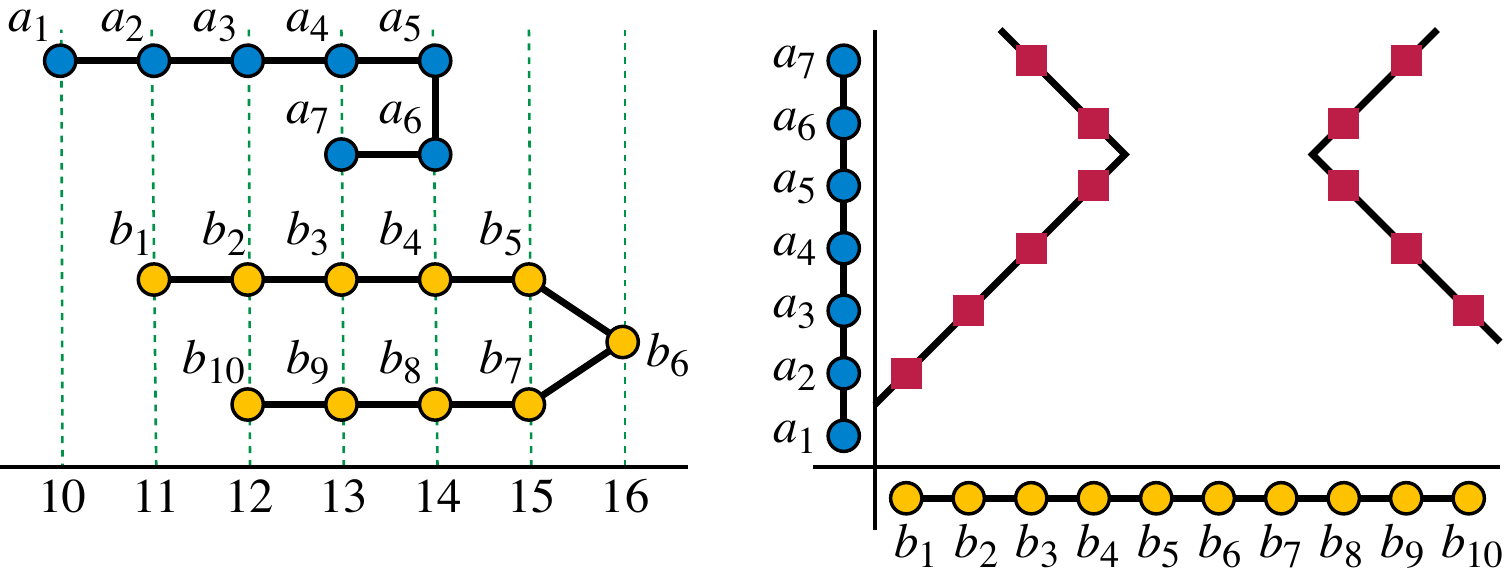}
\caption{Two branches $A$ and $B$ arranged by their distances from a locating point $s$ (left), and their indistinct set (of pairs not distinguished by $s$) plotted using the positions in the branches as Cartesian coordinates (right).}
\label{fig:indistinct}
\end{figure}

When plotted in two dimensions, with the position of $a$ in $A$ as one Cartesian coordinate and the position of $b$ in $B$ as the other, an indistinct set has the structure of $O(1)$ line segments with slopes $\pm 1$ (\autoref{fig:indistinct}). By rotating this coordinate system by $45^\circ$ we may use a more convenient coordinate system in which these segments are all horizontal or vertical, rather than diagonal. However, we must be careful when using this rotated system: only half of the integer points (the ones with even sums of coordinates) correspond to the integer points in the un-rotated system, which are the only points that can be members of an indistinct set.

\begin{observation}
Set $S$ is a valid locating set if, for every pair of branches $A$ and $B$, the different indistinct sets for the different points in $S$ have an empty intersection.
\end{observation}

\subsection{Stems}
Now, consider how the indistinct set of $s$, $A$ and $B$  changes as the position of $s$ varies along a third branch $C$ in the given graph.

\begin{definition}
We say that two indistinct sets are \emph{combinatorially equivalent} if there is a one-to-one correspondence between the diagonal segments of the two sets with the following properties:
\begin{itemize}
\item If $s$ is a diagonal of one indistinct set, then the corresponding diagonal in the other set has the same slope as $s$.
\item If $s$ and $t$ are two diagonals of one indistinct set that intersect each other, then the corresponding diagonals in the other set also intersect each other.
\item If $s$, $t$, and $u$ are three diagonals of one indistinct set, with $t$ and $u$ both intersecting~$s$, then the corresponding two intersections of diagonals in the other intersecting set have the same (northwest-to-southeast or northeast-to-southwest) ordering.
\end{itemize}
Combinatorial equivalence is an equivalence relation and we define the \emph{combinatorial structure} of an indistinct set to be its equivalence class in this equivalence relation.
\end{definition}

\begin{definition}
We define a \emph{stem} to be a maximal contiguous subset of a branch $C$ of the given graph $G$ within which the indistinct sets of all points $s$ in $C$ and all pairs $(A,B)$ of branches have the same combinatorial structure.
\end{definition}

\begin{lemma}
\label{lem:const-breaks}
For a given pair of branches $(A,B)$ and a third branch $C$, there are $O(1)$ positions along $C$ such that the indistinct set of a vertex $s$ of $C$ and the pair $(A,B)$ changes structure at that position.
\end{lemma}

\begin{proof}
The structure changes only at two types of position along $C$:
\begin{itemize}
\item positions where the shortest path from $s$ to an endpoint of $A$ or $B$ switches from going through one end of $C$ to going through the other, and
\item positions where the endpoints of the two branches $A$ and $B$ and the points of maximum distance from $s$ change their relative positions in the arrangement by distance from $s$.
\end{itemize}
There are two possibilities for the shortest path from $s$ to a given endpoint $v$ of $A$ or $B$: it must consist of a path in $C$ from $s$ to one endpoint of $C$ together with the shortest path in $G$ from that endpoint of $C$ to $v$. These two paths can change their ordering only once as $s$ moves along $C$. Therefore, each endpoint of $A$ or $B$ contributes at most four breakpoints of the first type.

Because it has only one breakpoint of the first type, each endpoint of $A$ or $B$ has a distance from $s$ that (as a function of the position of $s$ along $C$) is piecewise linear with only one breakpoint. By similar reasoning, the distance from $s$ to the farthest point within $A$ or $B$ is also piecewise linear with $O(1)$ breakpoints. Therefore, in the arrangement by distance, these points can exchange positions only $O(1)$ times.
\end{proof}

At all points of $C$ other than these, the indistinct set for $s$, $A$, and $B$ maintains the same combinatorial structure.  The positions of its segments either remain fixed as $s$ varies along the path, or they shift linearly with the position of $s$ along $C$.

\begin{lemma}
Every graph $G$ with $b$~branches has $O(b^3)$ stems.
\end{lemma}

\begin{proof}
There are $O(b^2)$ pairs of branches, each of which (by \autoref{lem:const-breaks}) contributes $O(1)$ breakpoints to branch $C$, so each branch has $O(b^2)$ stems and there are $O(b^3)$ stems in the whole graph.
\end{proof}

\subsection{The algorithm}

\begin{lemma}
\label{lem:locating-set-upper-bound}
The metric dimension of every graph with $b$ branches is $O(b)$.
\end{lemma}

\begin{proof}
A set $S$ that includes the endpoints of all branches and an interior point of each branch is certainly a valid locating set, and has $|S|=O(b)$.
\end{proof}

\begin{theorem}
The metric dimension of any graph with $n$ vertices and $b$ branches may be determined in time $O(n)+2^{O(b^3\log b)}\log n$.
\end{theorem}

\begin{proof}
We may assume without loss of generality that the graph is connected, for otherwise we could partition it into connected components and process each component separately.
Partitioning the graph into branches may be performed in time $O(n)$. After this step all shortest path computations in the given graph can be performed by instead using a weighted graph with $O(b)$ vertices and edges, in which each edge represents a branch of the original graph and is weighted by that branch's length. In particular, after partitioning the graph into branches, we may partition the branches into stems in total time~$b^{O(1)}$.

We search for locating sets of size $O(b)$ (according to \autoref{lem:locating-set-upper-bound})
by choosing nondeterministically the number of vertices in the locating set $S$, and the stem containing each vertex (but not the location of the vertex within the stem). There are $2^{O(b\log b)}$ possible choices of this type. This choice determines the combinatorial structure of each indistinct set.

Next, for each pair $(A,B)$ of branches (allowing $A=B$) and each member $s$ of the locating set (now associated with a specific stem but not placed at a particular vertex within that stem), we consider the line segments forming the indistinct sets for $s$, $A$ and $B$, in the rotated coordinate system for which these line segments are horizontal and vertical. For a given pair $(A,B)$ there are $O(b)$ line segments ($O(1)$ for each member of the locating set) and each line segment may be specified by the two Cartesian coordinate pairs for its endpoints. Rather than choosing these coordinate values numerically, we choose nondeterministically the sorted order of the $x$-coordinates and similarly the sorted order of the $y$-coordinates, allowing ties in our nondeterministic choices. In other words, separately for the $x$ and $y$ coordinates, we select a weak ordering of the segment endpoints, specifying for any two segment endpoints whether they have equal coordinate values or, if not, which one has a smaller coordinate value than the other.
We also choose nondeterministically the parity of each Cartesian coordinate.
Each of the $O(b^2)$ pairs of branches has $2^{O(b\log b)}$ choices for these orderings and parities, so there are $2^{O(b^3\log b)}$ possible nondeterministic choices overall.
For each such choice and each pair $(A,B)$ we verify that, if we can find a placement of the vertices of the locating set that gives rise to the chosen sorted orderings, then the intersection of the indistinct sets for $A$ and $B$ will not contain any integer points (in the un-rotated coordinate system).

To test whether two indistinct sets have a non-empty intersection, we test each pair of a line segment from one set and a line segment from the other set for an intersection.
Two horizontal line segments intersect each other if and only if they have the same $y$-coordinate and overlapping intervals of $x$-coordinates; a symmetric calculation is valid for two vertical line segments. A horizontal line segment intersects a vertical line segment if and only if the $y$-coordinate of the horizontal segment is within the range of $y$-coordinates of the vertical segment, the $x$-coordinate of the vertical segment is within the range of $x$-coordinates of the horizontal segment, and the parities of the coordinates of the two segments cause their crossing point to land on an integer point rather than on a half-integer point. In this way, the existence of an intersection point can be determined in time polynomial in $b$, using only the information about the sorted order and parities of coordinates that we have chosen nondeterministically.

When these nondeterministic choices find a collection of indistinct sets, and a sorted ordering of the features of those sets, for which every pair of branches has an empty intersection of indistinct sets, it remains to determine whether there exists a placement of each locating set vertex within its stem, in order to cause the indistinct set features to have the sorted orders that we have already chosen. Each ordering constraint between two features that are consecutive in one of the sorted orders translates directly to a linear constraint between the positions of two locating set vertices $s$ and $s'$ within their stems; therefore, the problem of finding positions that satisfy all of these constraints can be formulated and solved as an integer linear programming feasability problem, with $O(b)$ variables (the positions of the locating vertices on their stems) and $O(b^3)$ constraints (sorted orderings of $O(b)$ items for each of $O(b^2)$ pairs of branches, specified with numbers of $O(\log n)$ bits (the lengths of the stems).
By standard algorithms for low-dimensional integer linear programming problems, this problem can be solved in time $2^{O(b\log b)}\log n$.~\cite{Len-MOR-83,Kan-MOR-87,FraTar-Comb-87,Cla-JACM-95}.

The product of the numbers of nondeterministic choices made by the algorithm with the time for integer linear programming for each choice gives the stated time bound.
\end{proof}

\begin{corollary}
The metric dimension of a graph with max leaf number $\ell$ may be determined in time $O(n)+2^{O(\ell^6\log \ell)}\log n$.
\end{corollary}

\section{Conclusions}

We have shown that metric dimension is fixed-parameter tractable in the max leaf number, but our algorithms have time bounds that are too high to be practical. It would therefore be of interest to reduce this dependence, for instance to be singly-exponential in the max leaf number.

It would also be of interest to extend this method to stronger parameters. For instance, the fact that the metric dimension is relatively easy on trees~\cite{HarMel-AC-76} makes it plausible that, for general graphs, we could reduce the problem to one on the 2-core of the graph (the subgraph that remains after repeatedly removing degree-one vertices). The branch-count of the 2-core of any graph is proportional to the graph's cyclomatic number, so such a result would mean that the metric dimension could be computed in fixed-parameter tractable time in the cyclomatic number. Is this possible?

\subsection*{Acknowledgements}
This research was supported in part  by the National Science Foundation under grants 0830403 and 1217322, and by the Office of Naval Research under MURI grant N00014-08-1-1015.

{\raggedright
\bibliographystyle{abuser}
\bibliography{branch-count}}

\end{document}